\documentclass[sort&compress]{elsarticle}
\usepackage{amsmath,mathtools,amssymb, amsthm}
\usepackage{inputenc} 
\usepackage{graphicx}   
\usepackage{epstopdf}
\usepackage{subfigure}
\usepackage{color}
\usepackage{verbatim} 
\usepackage{booktabs}
\usepackage{bm}   
\usepackage{dutchcal}
\usepackage[left=3cm,right=3cm,top=2.5cm,bottom=2.5cm]{geometry}  
\usepackage{lineno,hyperref}
\usepackage{tablists} 
\usepackage{makecell} 
\usepackage{tikz}
\usepackage{hyperref}
\usepackage{pstricks}
\usepackage{appendix}
\usepackage{pstricks-add}
\usepackage{float}
\hypersetup{
	colorlinks=ture,
	linkcolor=blue,
	filecolor=gray,
	urlcolor=blue,
	citecolor=blue}
\numberwithin{equation}{section}
\newtheorem{definition}{Definition}[section]

\newtheorem{proposition}{Proposition}[section]

\begin{document}
	\begin{frontmatter}
		\title{Higher form Yang-Mills as higher BFYM theories}
		\author{Danhua Song \corref{cor1}}
		\ead{danhua_song@163.com}
		\author{Kai Lou}
		\author{Ke Wu}
		\ead{wuke@cnu.edu.cn}
		\author{Jie Yang}
		\ead{yangjie@cnu.edu.cn}
			\cortext[cor1]{Corresponding author.}
		\address{School of Mathematical Sciences, Capital Normal University, Beijing 100048, China}
		\date{}
		
		\begin{abstract}
	The YM theory has been generalized to 2YM and 3YM theories. Similarly, we generalize the BFYM theory to ``2BFYM" and ``3BFYM" theories. 	Then, we show that these higher BFYM theories can give the formulations of the corresponding higher form YM theories. Finally, we study the gauge symmetries of these higher BFYM theories.

		\end{abstract}
		
	\end{frontmatter}

\section{Introduction}

As an ordinary gauge theory, the Yang–Mills (YM) theory has many interesting aspects of physics, due to the role in the standard model in high energy physics \cite{WS}. 
In recent decades, higher-dimensional extended objects are considered as the basic constituents of matter and mediators of fundamental interactions. Therefore, many theories used to study point-like objects are generalized to higher versions. In particular, the higher gauge theory \cite{Baez.2010, Baez2005HigherGT, Bar, JBUS, Faria_Martins_2011, doi:10.1063/1.4870640, ACJ, FH}, as a generalization of the ordinary gauge theory, is considered to be the geometrically promising technique to describe the dynamics of the higher-dimensional extended objects. 
The basic fields studied in this theory are various higher gauge fields, which are locally given by higher differential forms with values in higher algebras, and are globally modeled by higher bundles with higher connections.  Moreover, the higher gauge theory has attracted considerable attentions in many branches of physics, such as the six-dimensional superconformal field theory \cite{Saemann:2013pca}, quantum gravity \cite{TRMV11, AMMV}, string theory \cite{USJ, U},  M-theory \cite{ HSG, SP, Fiorenza:2012mr} and so on.  
Within the higher gauge theory,  the electromagnetic theory has been also generalized to a $p$-form electromagnetic theory \cite{ 10.1007/BF01889624, HP, Kalb:1974yc}. Similarly, the YM theory has been generalized to the  2-form Yang-Mills (2YM) theory \cite{2002hep.th....6130B}  and 3-form Yang-Mills (3YM) theory \cite{sdh} for 2- and 3-dimensional extended objects, respectively.

The YM theory can be formulated as a deformation of the BF theory, called BF-YM (BFYM) theory \cite{ASCPCR, ASC, MMMAZ}. As a topological field theory, the BF theory plays a major role in the physics theory, such as the general relativity and quantum gravity  \cite{ASPCRJ, Celada:2016jdt}. 
Furthermore, the topological BF theory has been generalized to the topological 2BF \cite{ Martins:2010ry, FGHPEM, AMMA, AMMAMV, AMMAM, AMMA1} and 3BF theories \cite{Radenkovic:2019qme, TRMV,  MM1} within the framework of the higher gauge theory.
These higher BF theories can describe the correct dynamics of YM,  Dirac, Weyl, Klein-Gordon and Majorana fields coupled to Einstein-Cartan gravity. See \cite{Radenkovic:2020weu} for more details. 	
The developments of the higher YM and BF theories
motivate us to consider the higher counterparts of the BFYM theory by following the idea on the higher gauge theory.

Within the framework of the higher gauge theory, the BF and YM theories  were generalized to their higher counterparts by using the idea of a categorical ladder, which promotes the underlying algebraic structure from a group to a higher group. 
To be more precise, we sketch their mathematical structures by the steps of the categorical ladder in Table \ref{table 1}. As we shall see in section  \ref{sub2}, these higher groups have more abundant of properties.
\begin{center}
	\begin{table}[htp]
		\caption{The mathematical structures of (higher) gauge theories} \label{table 1}
		\setlength{\tabcolsep}{3mm}{
			\begin{tabular}{cccc}
				\specialrule{1pt}{5pt}{1pt} \specialrule{0.5pt}{1pt}{3pt}
				gauge theory & categorical structure& algebraic structure& linear structure   \\[1mm] 
				\specialrule{0.5pt}{2pt}{6pt} 
				BF, YM & Lie group & Lie group & Lie algebra  \\[1.5mm]
				\specialrule{0.5pt}{1pt}{7pt}  
				2BF, 2YM & Lie 2-group & Lie crossed module &
				\begin{tabular}[c]{@{}l@{}}\ \ \ \ differential  \\\ crossed module \end{tabular} \\[1.5mm]
				\specialrule{0.5pt}{1pt}{6pt} 
				3BF, 3YM & Lie 3-group& Lie 2-crossed module &  \begin{tabular}[c]{@{}l@{}} \ \ \ \ \ differential \\\
					2-crossed module \end{tabular}   \\[2.5mm]  
				\specialrule{0.5pt}{2pt}{1pt} \specialrule{1pt}{1pt}{0pt} 
		\end{tabular} } 
		
	\end{table}
\end{center}

In view of the BFYM theory connecting the YM and BF theories, three questions naturally arise:
\begin{itemize}
	\item whether higher BFYM theories exist by using the idea of the categorical ladder?
	\item how to construct the higher BFYM actions, which are classically equivalent to the higher YM theories?
	\item whether these higher BFYM theories are  gauge invariant under higher gauge transformations?
\end{itemize}

For the first question, we establish two new theories---called respectively 2BF-2-form YM (2BFYM) and 3BF-3-form YM (3BFYM) theories---which are constructed under the categorical generalizations.
For the second question, we prove that the 2BFYM and 3BFYM theories are respectively equivalent to the 2YM and 3YM theories on the classic level.
For the third question, we study the higher gauge symmetries of the 2BFYM and 3BFYM theories.
We list  our results (2BFYM and 3BFYM actions) and the known results in Table \ref{table 2}. We note that these actions of higher BFYM theories are nice generalizations of the known counterparts. 
\begin{center}
	\begin{table}[htp]
		\caption{The actions of (higher) BF, BFYM and YM theories}\label{table 2}
		\centering
		\setlength{\tabcolsep}{7mm}{
			\begin{tabular}{cc}
				\specialrule{1pt}{5pt}{1pt} \specialrule{0.5pt}{1pt}{3pt}
				theory   & action    \\[1mm] 
				\specialrule{0.5pt}{2pt}{6pt} 
				BF & $ S_{BF} = \displaystyle{\int_{M} \langle \overline{B}, F \rangle}$  \\[1.5mm]
				\specialrule{0.5pt}{1pt}{7pt}  
				BFYM & $S_{BFYM} = \displaystyle{\int_{M} (\langle \overline{B}, F \rangle + e^2 \langle \overline{B}, \ast \overline{B} \rangle)}$  \\[1.5mm]
				\specialrule{0.5pt}{1pt}{6pt} 
				YM& $S_{YM} = - \dfrac{\sigma}{4e^2} \displaystyle{\int_{M} \langle F, \ast F \rangle}$  \\[2.5mm]  
				\specialrule{0.5pt}{1pt}{6pt}
				2BF & $S_{2BF} = \displaystyle{\int_{M} (\langle \overline{B}, \mathcal{F} \rangle + \langle \overline{C}, \mathcal{G} \rangle )} $   \\[1.5mm]
				\specialrule{0.5pt}{1pt}{7pt}  
				2BFYM & $S_{2BFYM} = \displaystyle{\int_{M} (\langle \overline{B}, \mathcal{F} \rangle +  e^2 \langle \overline{B}, \ast \overline{B} \rangle + \langle \overline{C}, \mathcal{G} \rangle +  e^2 \langle \overline{C}, \ast \overline{C} \rangle )}$  \\[1.5mm]
				\specialrule{0.5pt}{1pt}{6pt}
				2YM& $S_{2YM} = - \dfrac{\sigma}{4e^2} \displaystyle{\int_{M} (\langle \mathcal{F} , \ast \mathcal{F} \rangle +\langle \mathcal{G} , \ast \mathcal{G} \rangle )}$  \\[2.5mm]  
				\specialrule{0.5pt}{1pt}{6pt}
				3BF & $S_{3BF} = \displaystyle{\int_{M} (\langle \overline{B}, \mathcal{F} \rangle + \langle \overline{C}, \mathcal{G} \rangle + \langle \overline{D}, \mathcal{H} \rangle)}$ \\[1.5mm]
				\specialrule{0.5pt}{1pt}{6pt}
				3BFYM  &$S_{3BFYM} = \displaystyle{\int_{M} (\langle \overline{B}, \mathcal{F} \rangle +  e^2 \langle \overline{B}, \ast \overline{B} \rangle + \langle \overline{C}, \mathcal{G} \rangle +  e^2 \langle \overline{C}, \ast \overline{C} \rangle+  }$\\[3mm]
				\qquad  & $\langle \overline{D}, \mathcal{H} \rangle + e^2 \langle \overline{D}, \ast \overline{D} \rangle)$ \\[1.5mm]
				\specialrule{0.5pt}{1pt}{6pt}
				3YM& $S_{3YM} = - \dfrac{\sigma}{4e^2} \displaystyle{\int_{M} (\langle \mathcal{F} , \ast \mathcal{F}\rangle + \langle \mathcal{G} , \ast \mathcal{G}\rangle + \langle \mathcal{H}, \ast \mathcal{H} \rangle) }$   \\[1.5mm]
				\specialrule{0.5pt}{2pt}{1pt} \specialrule{1pt}{1pt}{0pt} 
		\end{tabular} } 
	\end{table}
\end{center}

The layout of the paper is summarized in the following outline:
\begin{enumerate}[1. ]
	\item In section \ref{sub2}, we review the relevant algebraic tools involved in the description of higher YM and BFYM theories, and generalize the $G$-invariant bilinear forms from a differential crossed module to a nice differential 2-crossed module. Besides, we  introduce seven key maps for these higher algebras, which are crucial for the construction of the higher gauge theories. 
	\item In section \ref{sub3}, we define the differential forms valued in the higher algebras, and we prove four propositions and present seven related maps for the spaces of Lie algebra valued differential forms, which will play major roles in the our calculations.
	\item In section \ref{YMBFYM}, we review an important result that the YM theory is classically equivalent to the BFYM theory, and prove the gauge symmetry of the BF and BFYM theories.
	\item In section \ref{HBFYM}, we present the two main results of this paper.  Firstly, in subsection \ref{2BFYM} we introduce the 2BF and 2YM theories, and construct the 2BFYM action, which can define the dynamics of the 2YM theory. We prove the gauge symmetry of 2BFYM theory under the thin and fat gauge transformations, respectively.
	Secondly, in subsection \ref{3BFYM}, we introduce the 3BF and 3YM theories, and 
	establish the 3BFYM action, which is classically equivalent to the 3YM theory. Similarly, we also prove the gauge symmetry of 3BFYM theory under three gauge transformations, respectively.
\end{enumerate}

\section{Preliminaries}\label{sub2}
\subsection{Lie (2-)crossed modules and differential (2-)crossed modules} 
In this section, we first give the necessary definitions of Lie crossed modules, Lie 2-crossed modules, and their infinitesimal versions. Then, we are going to introduce the $G$-invariant forms in these algebraic structures, which will be needed for the construction of the higher BFYM actions. 
For more details, refer the readers to  \cite{Faria_Martins_2011, 2002hep.th....6130B, Beaz, Crans, Brown, Kamps20022groupoidEI, Mutlu1998FREENESSCF, Roberts2007TheIA}.

\paragraph{Lie pre-crossed modules}
A Lie pre-crossed module $\left(H,G;\alpha,\vartriangleright \right)$ consists of two Lie groups $H$, $G$ together with a Lie group map $\alpha: H  \longrightarrow G$ and a smooth left action $\vartriangleright$ of G on H by automorphisms such that:
\begin{equation}
	\alpha \left(g \vartriangleright h\right) = g \alpha \left(h\right) g^{-1},
\end{equation}
for each $g \in G $ and $h \in H$, and  a Peiffer commutator $\lbrack\lbrack \cdot , \cdot \rbrack\rbrack : H \times H \longrightarrow H $ defined by
\begin{equation}
	\lbrack\lbrack h , h' \rbrack\rbrack = h h' h^{-1}\left(\alpha \left(h\right) \vartriangleright h'^{-1}\right),
\end{equation}
for each $h, h' \in H$.

\paragraph{Lie crossed modules}
A Lie pre-crossed module $\left(H,G;\alpha,\vartriangleright \right)$ is said to be  a Lie crossed module (or a strict Lie 2-group), if  its Peiffer commutator is trivial, i.e.
\begin{equation}
\alpha \left(h\right) \vartriangleright h' = h h' h^{-1}.
\end{equation}
Besides, the map $\left(h_1 , h_2\right) \in H \times H \longrightarrow \lbrack\lbrack h , h' \rbrack\rbrack \in H $, called the Peiffer pairing, is G-equivariant, i.e.
\begin{equation}
g \vartriangleright \lbrack\lbrack h_1 , h_2 \rbrack\rbrack = \lbrack\lbrack g \vartriangleright h_1 , g \vartriangleright h_2 \rbrack\rbrack ,
\end{equation}
for each $ h_1 , h_2 \in H$ and $g \in G$. Moreover, $\lbrack\lbrack h_1 , h_2 \rbrack\rbrack =1_H $ if either $h_1$ or $h_2$ is $1_H$. A simple example of the Lie crossed module is $G = H = U(N) $ with $\alpha $ the identity map and the $ \vartriangleright $ the adjoint action.

There is an infinitesimal version of the Lie pre-crossed module.   
In order not to introduce additional notations, we use the same letters $\alpha$ and $\vartriangleright$ for counterparts in the infinitesimal version. 
The convention also applies to the Lie (2-)crossed modules.

\paragraph{Differential pre-crossed modules}
A differential pre-crossed module $(\mathcal{h}, \mathcal{g}; \alpha, \vartriangleright)$ consists of two Lie algebras $\mathcal{h}$, $\mathcal{g}$ together with a Lie algebra map $\alpha : \mathcal{h} \longrightarrow \mathcal{g}$ and a left action $\vartriangleright$ of $\mathcal{g}$ on $\mathcal{h}$ by derivations such that
\begin{equation}\label{XY}
\alpha(X\vartriangleright Y ) =\left[ X, \alpha(Y) \right],
\end{equation}
for each $X \in \mathcal{g}$ and $ Y \in \mathcal{h}$,
and a Peiffer commutator $\lbrack\lbrack \  ,\  \rbrack\rbrack : \mathcal{h} \times \mathcal{h} \longrightarrow \mathcal{h} $ defined by
\begin{equation}\label{jkh}
\lbrack\lbrack Y , Y'\rbrack\rbrack = \left[Y,Y'\right] - \alpha(Y)\vartriangleright Y',
\end{equation}
for each $Y, Y' \in \mathcal{h}$.

\paragraph{Differential crossed modules}
A differential pre-crossed module $(\mathcal{h}, \mathcal{g}; \alpha, \vartriangleright)$ is said to be  a differential crossed module (or a strict Lie 2-algebra), if its Peiffer commutator vanish, i.e.
\begin{equation}\label{ YY'}
\alpha(Y)\vartriangleright Y'=\left[Y,Y'\right].
\end{equation}
Moreover, the map $(Y_1,Y_2) \in \mathcal{h} \times \mathcal{h} \longrightarrow 	\lbrack\lbrack Y_1 , Y_2\rbrack\rbrack  \in \mathcal{h}$, called the Peiffer pairing, is $\mathcal{g}$-equivariant, i.e.
\begin{equation}
X \vartriangleright \lbrack\lbrack Y_1 , Y_2\rbrack\rbrack = \lbrack\lbrack X \vartriangleright Y_1 , Y_2 \rbrack\rbrack +\lbrack\lbrack  Y_1 , X \vartriangleright Y_2 \rbrack\rbrack,
\end{equation}
and the map $(X,Y)\in \mathcal{g} \times \mathcal{h} \longrightarrow X \vartriangleright Y \in \mathcal{h}$ is bilinear, i.e.
\begin{align}
X \vartriangleright \left[Y_1 ,  Y_2\right]  &= \left[ X \vartriangleright Y_1 , Y_2\right] +\left[  Y_1 , X \vartriangleright Y_2\right],\label{XY_1Y_2}\\
\left[X_1, X_2\right] \vartriangleright Y &= X_1 \vartriangleright (X_2 \vartriangleright Y) - X_2 \vartriangleright(X_1 \vartriangleright Y).\label{X_1X_2Y}
\end{align}
The differential version of the above simple example is evidently $\mathcal{g} = \mathcal{h} = \mathcal{u} (N) $ with $\vartriangleright $ being the adjoint action and $ \alpha $ the identity map. 

\paragraph{Lie 2-crossed module}
A Lie 2-crossed module $(L, H, G;\beta, \alpha, \vartriangleright, \left\{ \ ,\  \right\})$ is given by a complex of Lie groups
$$L \stackrel{\beta}{\longrightarrow}H \stackrel{\alpha}{\longrightarrow}G$$
together with a smooth left action $\vartriangleright$ by automorphisms of $G$ on $L$ and $H$ (and on $G$ by conjugation), i.e.
\begin{equation}
g \vartriangleright (e_1 e_2)=(g \vartriangleright e_1)(g \vartriangleright e_2), \ \ \ (g_1 g_2)\vartriangleright e = g_1 \vartriangleright (g_2 \vartriangleright e),
\end{equation}
for each $g, g_1, g_2\in G, e, e_1, e_2\in H $ or $L$, and a $G$-equivariant smooth function $ \left\{ \ ,\  \right\} :H \times H \longrightarrow L $, called the Peiffer lifting, such that
\begin{equation}
g \vartriangleright \left\{ h_1, h_2 \right\} = \left\{g \vartriangleright h_1, g \vartriangleright h_2\right\},
\end{equation}
for each $h_1, h_2\in H$. Moreover, there is a left action of $H$ on $L$ by automorphisms $\vartriangleright' $ which is defined by
\begin{equation}
h \vartriangleright' l = l \left\{ \beta (l)^{-1}, h \right\}.
\end{equation}
This action $\vartriangleright'$ together with the homomorphism $ \beta$ defines a crossed module $(L, H; \beta, \vartriangleright')$. 
The Lie 2-crossed module satisfies many properties, but we will not give the detailed introduction here. For more details, see \cite{Roberts2007TheIA}.

\paragraph{Differential 2-crossed modules}
A differential 2-crossed module$(\mathcal{l},\mathcal{h}, \mathcal{g}; \beta, \alpha,\vartriangleright, \left\{ , \right\})$ is given by a complex of Lie algebras:
$$\mathcal{l}\stackrel{\beta}{\longrightarrow} \mathcal{h} \stackrel{\alpha}{\longrightarrow} \mathcal{g},$$
together with left action $\vartriangleright $ by derivations of $\mathcal{g}$ on $\mathcal{l}, \mathcal{h}, \mathcal{g}$ (on the latter by the adjoint representation), and a $\mathcal{g}$-equivariant bilinear map $\left\{ \ ,\  \right\}:\mathcal{h} \times \mathcal{h} \longrightarrow \mathcal{l}$, the Peiffer lifting, such that
\begin{equation}\label{12}
X\vartriangleright \left\{ Y_1,Y_2\right\} = \left\{X\vartriangleright Y_1,Y_2\right\} + \left\{Y_1, X\vartriangleright Y_2\right\},
\end{equation}
for each $X \in \mathcal{g}$ and $Y_1,Y_2 \in \mathcal{h}$. 

The differential 2-crossed module satisfies the following properties:
\begin{enumerate}
\item 	$\mathcal{l}\stackrel{\beta}{\longrightarrow} \mathcal{h} \stackrel{\alpha}{\longrightarrow} \mathcal{g}$ is a complex of $\mathcal{g}$-modules satisfying $\alpha \circ \beta=0$;

\item $\beta \left\{Y_1,Y_2\right\} =\lbrack\lbrack Y_1 , Y_2\rbrack\rbrack $, for each $Y_1, Y_2\in \mathcal{h}$, where $\lbrack\lbrack Y_1 , Y_2\rbrack\rbrack=\left[ Y_1, Y_2\right] - \alpha (Y_1) \vartriangleright Y_2$;

\item $ \left[Z_1, Z_2\right]= \left\{ \beta (Z_1), \beta (Z_2)\right\}$;

\item $\left\{ \left[Y_1 ,Y_2\right], Y_3 \right\}= \alpha (Y_1) \vartriangleright \left\{Y_2,Y_3\right\}+\left\{Y_1,\left[Y_2,Y_3\right]\right\}-\alpha (Y_2)\vartriangleright \left\{Y_1,Y_3\right\}-\left\{Y_2,\left[Y_1,Y_3\right]\right\}$. This is the same as :
$$\left\{\left[Y_1,Y_2\right],Y_3\right\}= \left\{\alpha (Y_1)\vartriangleright Y_2, Y_3\right\} - \left\{ \alpha (Y_2)\vartriangleright Y_1, Y_3\right\} - \left\{Y_1, \beta \left\{ Y_2, Y_3\right\}\right\} + \left\{Y_2, \beta\left\{Y_1,Y_2 \right\}\right\};$$

\item $ \left\{Y_1,\left[Y_2,Y_3\right]\right\}= \left\{ \beta\left\{Y_1,Y_2 \right\},Y_3  \right\}-\left\{ \beta\left\{Y_1,Y_3 \right\},Y_2 \right\}$;

\item $\left\{ \beta (Z) , Y \right\} +\left\{Y, \beta (Z)\right\} =- (\alpha (Y)\vartriangleright Z)$, 
\end{enumerate}
for each $Y, Y_1,Y_2, Y_3 \in \mathcal{h} $ and $Z, Z_1,Z_2 \in \mathcal{l}$.

Analogously to the Lie 2-crossed module case, there is a left action of $\mathcal{h}$ on $\mathcal{l}$ which is defined by
\begin{equation}\label{YZ}
Y\vartriangleright'Z= -\left\{ \beta(Z),Y\right\}.
\end{equation}
This together with the homomorphism $\beta$ defines a differential crossed module $(\mathcal{l}, \mathcal{h}; \beta, \vartriangleright')$.

\subsection{$G$-invariant bilinear forms in the differential (2-)crossed module}\label{G2}

In this paper, we consider a nice 2-crossed module $(L, H, G;\beta, \alpha, \vartriangleright, \left\{ \ ,\  \right\})$, where $(H, G; \alpha, \vartriangleright)$ is a Lie crossed module.
Similarly, a differential 2-crossed module $(\mathcal{l},\mathcal{h},\mathcal{g};\beta,\alpha,\vartriangleright,\left\{ \ ,\  \right\})$ is nice if and only if $(\mathcal{h}, \mathcal{g}; \alpha, \vartriangleright)$ is a  differential crossed module.
Thus, the bilinear forms in the differential crossed module will be included in that of  the differential 2-crossed module. 
We just need to study the bilinear forms in the latter.

For a nice Lie 2-crossed module $(L, H, G;\beta, \alpha, \vartriangleright, \left\{ \ ,\  \right\})$ and the corresponding infinitesimal version $(\mathcal{l},\mathcal{h},\mathcal{g};\beta,\alpha,\vartriangleright,\left\{ \ ,\  \right\})$, there are two mixed relations \cite{Martins:2010ry},
induced by the action via $\vartriangleright$ on $H$ and also denoted by $\vartriangleright$,
\begin{align}
\alpha (g \vartriangleright Y) = g \alpha(Y) g^{-1},\ \ \  
\alpha (h) \vartriangleright Y= h Y h^{-1},\label{mixed1}
\end{align}
for each $g \in G, h \in H$ and $Y \in \mathcal{h}$.
Thinking of the Lie crossed module $(L, H; \beta, \vartriangleright')$, there are also two mixed relations corresponding to the differential crossed module $(\mathcal{l}, \mathcal{h}; \beta, \vartriangleright')$, induced by the action via $\vartriangleright'$ on $L$ and also denoted by $\vartriangleright'$,
\begin{align}
\beta (h \vartriangleright' Z) = h \beta (Z) h^{-1},\ \ \ 
\beta (l) \vartriangleright' Z = l Z l^{-1},\label{mixed2}
\end{align}
for each $h \in H, l \in L$ and $ Z \in \mathcal{l}$.

\begin{definition}[$G$-invariant non-degenerate symmetric bilinear  forms]
Let us consider a Lie 2-crossed module $(L,H,G;\beta,\alpha,\vartriangleright,\left\{ \ ,\  \right\})$ and $(\mathcal{l},\mathcal{h},\mathcal{g};\beta,\alpha,\vartriangleright,\left\{ \ ,\  \right\})$ be the associated differential 2-crossed module. A $G$-invariant non-degenerate symmetric bilinear form in $(\mathcal{l},\mathcal{h},\mathcal{g};\beta,\alpha,\vartriangleright,\left\{ \ ,\  \right\})$ is given by a triple of non-degenerate symmetric bilinear forms $\langle \ , \ \rangle_\mathcal{g}$ in $\mathcal{g}$ , $\langle \ , \ \rangle_\mathcal{h}$ in $\mathcal{h}$ and $\langle \ , \ \rangle_\mathcal{l}$ in $\mathcal{l}$ such that
\begin{enumerate}
\item  $\langle \ , \ \rangle_\mathcal{g}$ is $G$-invariant, i.e.
$$ \langle gXg^{-1}, gX'g^{-1} \rangle_\mathcal{g}= \langle X,X'\rangle_\mathcal{g}, \ \ \ \ \ \ \ \ \forall g \in G, \ X, X' \in \mathcal{g};$$
\item  $\langle \ , \ \rangle_\mathcal{h}$ is $G$-invariant, i.e.
$$ \langle g\vartriangleright Y,  g\vartriangleright Y' \rangle_\mathcal{h}= \langle Y,Y'\rangle_\mathcal{h}, \ \ \ \ \ \ \ \forall g\in G,\ Y,Y'\in \mathcal{h};$$
\item  $\langle \ , \ \rangle_\mathcal{l}$ is $G$-invariant, i.e.
$$ \langle g\vartriangleright Z,  g\vartriangleright Z' \rangle_\mathcal{l}= \langle Z,Z'\rangle_\mathcal{l}, \ \ \ \ \ \ \ \forall g\in G,\ Z,Z'\in \mathcal{l}.$$
\end{enumerate}
\end{definition}

Note that $\langle \ , \ \rangle_\mathcal{h}$ is necessarily $H$-invariant. Since
$$\langle h Y h^{-1}, h Y' h^{-1} \rangle_{\mathcal{h}} = \langle \alpha(h) \vartriangleright Y, \alpha(h) \vartriangleright Y' \rangle_{\mathcal{h}} = \langle Y, Y'\rangle_{\mathcal{h}},$$
for each $h \in H$,
by using the mixed relation \eqref{mixed1}.
In the case when the Peiffer lifting or the map $\beta$ is trivial consequently, $\langle \ , \ \rangle_\mathcal{l}$ is $H$-invariant, i.e.
$$ \langle h\vartriangleright' Z, h \vartriangleright' Z' \rangle_\mathcal{l} = \langle Z,Z'\rangle_\mathcal{l}.$$
Besides, $\langle \ , \ \rangle_\mathcal{l}$ is necessarily $L$-invariant based on the $H$-invariance of $\langle \ , \ \rangle_\mathcal{l}$, i.e.
$$\langle lZl^{-1}, lZ'l^{-1} \rangle_\mathcal{l} = \langle \beta (l) \vartriangleright' Z, \beta (l) \vartriangleright' Z' \rangle_\mathcal{l} = \langle Z,Z'\rangle_\mathcal{l},$$
for each $l\in L$, by using the mixed relation \eqref{mixed2}. See \cite{Radenkovic:2019qme} for more details.

There are no compatibility conditions between the symmetric bilinear forms  $\langle \ , \ \rangle_\mathcal{g}$ ,$\langle \ , \ \rangle_\mathcal{h}$ and $\langle \ , \ \rangle_\mathcal{l}$. Any representation of $G$ can be made unitary if $G$ is a compact group \cite{WFJH}. There exists a rich class of examples of the Lie crossed module in which $G$ is compact and $H$ is a non-abelian group, by constructing them from chain complexes of vector spaces\cite{Martins:2010ry}. In likewise, there is also a good Lie 2-crossed module with the desired properties constructed from chain complexes of vector spaces  \cite{ Faria_Martins_2011, Martins:2010ry, Kamps20022groupoidEI}. 

These invariance conditions imply that:
\begin{align}
\langle [X,X'], X''\rangle_\mathcal{g}&=-\langle X',[X,X'']\rangle_\mathcal{g},\label{XXX}\\
\langle [Y,Y'], Y''\rangle_\mathcal{h}&=-\langle Y',[Y,Y'']\rangle_\mathcal{h},\label{YYY}\\
\langle [Z,Z'], Z''\rangle_\mathcal{l}&=-\langle Z',[Z,Z'']\rangle_\mathcal{l}.\label{ZZZ}
\end{align}
Besides, there are some important maps used in this paper.
Define two bilinear antisymmetric maps $\sigma:\mathcal{h} \times \mathcal{h} \longrightarrow \mathcal{g}$ by 
\begin{align}
\langle \sigma(Y,Y'), X\rangle_\mathcal{g}=-\langle Y, X\vartriangleright Y'\rangle_\mathcal{h},
\end{align}
and   $\kappa:\mathcal{l} \times \mathcal{l} \longrightarrow \mathcal{g}$ by 
\begin{align}
\langle \kappa(Z,Z'), X\rangle_\mathcal{g}=-\langle Z, X\vartriangleright Z'\rangle_\mathcal{l},
\end{align}
for each $X \in \mathcal{g}$, $Y, Y' \in \mathcal{h}$ and $Z, Z' \in \mathcal{l}$.
Due to $\sigma(Y',Y)=-\sigma(Y,Y')$ and $\kappa(Z',Z)=-\kappa(Z,Z')$, we have
\begin{align}
\langle Y, X\vartriangleright Y'\rangle_\mathcal{h}&=-\langle Y', X\vartriangleright Y \rangle_\mathcal{h}=-\langle X\vartriangleright Y, Y' \rangle_\mathcal{h},\label{YXY'}\\
\langle Z, X\vartriangleright Z'\rangle_\mathcal{l}&=-\langle Z', X\vartriangleright Z \rangle_\mathcal{l}=-\langle X\vartriangleright Z, Z' \rangle_\mathcal{l}.\label{ZXZ'}
\end{align}
Then, define two bilinear maps $\eta_1: \mathcal{l} \times \mathcal{h} \longrightarrow \mathcal{h}$ and $\eta_2: \mathcal{l} \times \mathcal{h} \longrightarrow \mathcal{h}$ by
\begin{align}\label{YY'Z}
\langle \left\{ Y, Y' \right\}, Z \rangle_\mathcal{l} = -\langle Y', \eta_1(Z, Y) \rangle_\mathcal{h} = -\langle Y, \eta_2(Z, Y') \rangle_\mathcal{h},
\end{align}
and a map $\alpha^*:\mathcal{g} \longrightarrow \mathcal{h}$ in \cite{2002hep.th....6130B} by
\begin{align}
\langle Y, \alpha^*(X)\rangle_\mathcal{h}=\langle \alpha(Y), X\rangle_\mathcal{g},
\end{align}
and a map $\beta^*:\mathcal{h} \longrightarrow \mathcal{l}$ by 
\begin{align}
\langle Z, \beta^*(Y)\rangle_\mathcal{l}=\langle \beta(Z), Y\rangle_\mathcal{h}.
\end{align}
Finally, define a trilinear map $\theta: \mathcal{h} \times \mathcal{h} \times \mathcal{l} \longrightarrow \mathcal{g} $ satisfying
\begin{align}
\langle \theta (Y, Y', Z) , X \rangle_\mathcal{g} = - \langle Z, \left\{X \vartriangleright Y ,  Y'\right\} \rangle_\mathcal{l}.
\end{align}
These maps will play important roles in our calculations. See \cite{TRMV, Radenkovic:2019qme} for more properties.

\section{Lie algebra valued differential forms} \label{sub3}
In this section we will mostly follow the original notations and definitions of \cite{doi:10.1063/1.4870640}.
Given a Lie algebra $\mathcal{g}$, there is a vector space $\Lambda^k (M,\mathcal{g})$
of $\mathcal{g}$-valued differential $k$-forms on $M$. For $A\in \Lambda^k (M, \mathcal{g})$, we write $A=\sum\limits_{a}A^a X_a $ where $A^a$ is a scalar differential $k$-form and  $X_a$ is an element in $\mathcal{g}$. 
Here, we assume $\mathcal{g}$ to be a matrix Lie algebra. Then, we have $\left[X, X'\right]=XX'-X'X$ for each $X, X'\in \mathcal{g}$.
For $A=\sum\limits_{a}A^a X_a \in \Lambda^{k_1} (M, \mathcal{g}) $, $A'=\sum\limits_{b}A'^b X_b \in \Lambda^{k_2} (M, \mathcal{g}) $, 
define
\begin{align*}
A \wedge A' :=\sum\limits_{a,b}A^a \wedge A'^b X_a X_b, \ \ \   dA := \sum\limits_{a} dA^a X_a, \\
A \wedge^{\left[\ ,\  \right]} A' :=\sum\limits_{a,b}A^a \wedge A'^b \left[X_a, X_b\right], 
\end{align*}
then there is an identity
$$A\wedge^{\left[\ ,\  \right]} A'=A\wedge A'-(-1)^{k_1 k_2}A'\wedge A .$$
Similar arguments apply to $\mathcal{h}$ and $\mathcal{l}$.
Besides, for $B=\sum\limits_{a} B^a Y_a \in \Lambda^{t_1}(M,\mathcal{h})$, $B'=\sum\limits_{b} B'^b Y_b \in \Lambda^{t_2}(M,\mathcal{h})$ with $Y_a$, $Y_b \in \mathcal{h}$, define
\begin{align*}
A\wedge^{\vartriangleright }B :=& \sum\limits_{a,b} A^a \wedge B^b X_a \vartriangleright Y_b\ ,\\
B\wedge^{\left\{ \ ,\  \right\}} B':=& \sum\limits_{a,b} B^a \wedge B'^b \left\{Y_a,Y_b\right\},\\
B\wedge^{\lbrack\lbrack \ ,\  \rbrack\rbrack} B':=& \sum\limits_{a,b}B^a \wedge B'^b \lbrack\lbrack Y_a , Y_b \rbrack\rbrack\ ,\\
\alpha (B) :=& \sum\limits_{a}B^a \alpha(Y_a).
\end{align*}
For $C=\sum\limits_{a}C^a Z_a \in \Lambda^{q_1} (M, \mathcal{l})$, $C'=\sum\limits_{b}C'^b Z_b \in \Lambda^{q_2} (M, \mathcal{l})$, define
\begin{align*}
A\wedge^\vartriangleright C :=\sum\limits_{a,b} A^a \wedge C^b X_a \vartriangleright Z_b, \ \  \beta(C):=\sum\limits_{a}C^a \beta(Z_a),\\
B\wedge^{\vartriangleright'}C:=\sum\limits_{a,b}B^a \wedge C^b Y_a\vartriangleright' Z_b\
\end{align*}
where $Y_a \vartriangleright' Z_b = -\left\{\beta(Z_b), Y_a\right\}$ by using \eqref{YZ}.

The following propositions will give some properties for the Lie algebra valued differential forms, corresponding to the identities of the differential (2-)crossed module.
\begin{proposition}
For $A\in \Lambda^k (M, \mathcal{g})$, $A_1 \in \Lambda^{k_1} (M, \mathcal{g})$, $A_2 \in \Lambda^{k_2} (M, \mathcal{g})$, $B \in \Lambda^t (M, \mathcal{h})$, $B_1 \in \Lambda^{t_1} (M, \mathcal{h})$ and $ {B_2} \in \Lambda^{t_2} (M, \mathcal{h})$, have
\begin{enumerate}
\item $ \alpha (A\wedge^\vartriangleright B)=A\wedge^{\left[\ ,\ \right]}\alpha(B)$;

\item $\alpha (B_1)\wedge^\vartriangleright B_2 = B_1\wedge^{\left[\ ,\ \right]}B_2$;

\item $A\wedge^\vartriangleright (B_1\wedge^{\left[\ ,\ \right]} B_2)=(A\wedge^\vartriangleright B_1)\wedge^{\left[\ ,\ \right]}B_2+ (-1)^{kt_1}B_1\wedge^{\left[\ ,\ \right]}(A\wedge^\vartriangleright B_2)$;

\item $(A_1\wedge^{\left[\ ,\ \right]}A_2)\wedge^\vartriangleright B=A_1 \wedge^\vartriangleright(A_2\wedge^\vartriangleright B)+(-1)^{k_1 k_2 +1} A_2\wedge^\vartriangleright(A_1\wedge^\vartriangleright B)$.
\end{enumerate}
\end{proposition}
\begin{proof}
We can get these identities easily by using \eqref{XY}, \eqref{ YY'}, \eqref{XY_1Y_2} and \eqref{X_1X_2Y}.
\end{proof}

\begin{proposition}
\begin{enumerate}
\item For $A\in \Lambda^k (M,\mathcal{g})$,$A'\in  \Lambda^{k'}(M,\mathcal{g})$ and $C\in \Lambda^* (M,\mathcal{l})$,
\begin{align}
\beta(A\wedge^\vartriangleright C)&=A\wedge^\vartriangleright \beta(C),\label{AC}\\
A\wedge^\vartriangleright A' &= A\wedge A'+(-1)^{kk'+1} A'\wedge A.
\end{align}

\item For $A\in \Lambda^k (M,\mathcal{g})$,$B_1 \in \Lambda^{t_1}(M,\mathcal{h})$,$B_2 \in \Lambda^{t_2} (M,\mathcal{h})$ and $W \in \Lambda^* (M,\mathcal{w})$ $(\mathcal{w}=\mathcal{g},\mathcal{h},\mathcal{l})$,
\begin{align}
d(A\wedge^\vartriangleright W)&=dA\wedge^\vartriangleright W + (-1)^k A \wedge^\vartriangleright dW,\label{AW}\\
d(B_1\wedge^{\left\{ \ ,\  \right\}}B_2)&=dB_1\wedge^{\left\{ \ ,\  \right\}}B_2+(-1)^{t_1} B_1 \wedge^{\left\{ \ ,\  \right\}}dB_2, \label{BB}\\
A\wedge^\vartriangleright (B_1\wedge^{\left\{ \ ,\  \right\}}B_2)&=(A\wedge^\vartriangleright B_1)\wedge^{\left\{ \ ,\  \right\}}B_2 + (-1)^{kt_1}B_1\wedge^{\left\{ \ ,\  \right\}}(A\wedge^\vartriangleright B_2),\label{ABB}
\end{align}
\end{enumerate}
\end{proposition}
\begin{proof}
For a rigorous proof of this theorem, refer the readers to \cite{doi:10.1063/1.4870640}.
\end{proof}

Furthermore, we have $G$-invariant forms in $\Lambda^k(M, \mathcal{g})$, $\Lambda^t(M, \mathcal{h})$ and $\Lambda^q(M, \mathcal{l})$ induced by $\langle \ ,\  \rangle_\mathcal{g}$, $\langle\ ,\  \rangle_\mathcal{h}$ and $\langle\ ,\  \rangle_\mathcal{l}$ and we denote them by $\langle \ , \ \rangle$ defined as follows
$$\langle A, A'\rangle := \sum\limits_{a,b}A^a \wedge A'^b \langle X_a, X_b\rangle_\mathcal{g}, \quad \langle B, B'\rangle := \sum\limits_{a,b}B^a \wedge B'^b \langle Y_a, Y_b\rangle_\mathcal{h}, $$
$$ \langle C, C'\rangle := \sum\limits_{a,b}C^a \wedge C'^b \langle Z_a, Z_b\rangle_\mathcal{l}.$$
Then, we have
$$ \langle A, A'\rangle=(-1)^{k_1 k_2} \langle A', A\rangle, \qquad \langle B, B'\rangle=(-1)^{t_1 t_2} \langle B', B\rangle, $$$$\qquad \langle C, C'\rangle=(-1)^{q_1 q_2} \langle C', C\rangle,$$ 
using the symmetries of $\langle \ ,\  \rangle_\mathcal{g}$, $\langle \ ,\  \rangle_\mathcal{h}$ and $\langle \ ,\  \rangle_\mathcal{l}$.

The following propositions will give some properties for the $G$-invariant forms of the Lie algebra valued differential forms.
\begin{proposition}
For $A_i \in \Lambda^{k_i}(M,\mathcal{g})$, $B_j\in \Lambda^{t_j}(M,\mathcal{h})$ and $C_k\in \Lambda^{q_k}(M,\mathcal{l}) (i, j, k = 1, 2, 3)$, we have 
\begin{align}
\langle A_1 \wedge^{[ \ ,\ ]}A_2, A_3\rangle &= (-1)^{k_1k_2+1}\langle A_2, A_1 \wedge^{[\ ,\ ]}A_3 \rangle,\label{A_1A_2A_3}\\
\langle B_1 \wedge^{[\ ,\ ]}B_2, B_3\rangle &= (-1)^{t_1t_2+1}\langle B_2, B_1 \wedge^{[\ ,\ ]}B_3 \rangle,\\
\langle C_1 \wedge^{[\ ,\ ]}C_2, C_3\rangle &= (-1)^{q_1q_2+1}\langle C_2, C_1 \wedge^{[\ ,\ ]}C_3 \rangle.
\end{align}
\end{proposition}
\begin{proof}
It is straightforward to get these identities by using \eqref{XXX}, \eqref{YYY} and \eqref{ZZZ}.
\end{proof}

\begin{proposition}
For $A \in \Lambda^{k}(M,\mathcal{g})$, $B_1 \in \Lambda^{t_1}(M, \mathcal{h})$, $B_2 \in \Lambda^{t_2}(M, \mathcal{h})$, $C_1 \in \Lambda^{q_1}(M, \mathcal{l})$ and $C_2 \in \Lambda^{q_2}(M, \mathcal{l})$,  we have
\begin{align}
\langle B_1, A\wedge^\vartriangleright B_2\rangle &=(-1)^{t_2(k+t_1)+kt_1+1}\langle B_2, A\wedge^\vartriangleright B_1 \rangle =(-1)^{kt_1+1}\langle A\wedge^\vartriangleright B_1, B_2\rangle,\label{B_1AB_2}\\
\langle C_1, A\wedge^\vartriangleright C_2\rangle &=(-1)^{q_2(k+q_1)+kq_1+1}\langle C_2, A\wedge^\vartriangleright C_1 \rangle =(-1)^{kq_1+1}\langle A\wedge^\vartriangleright C_1,  C_2\rangle.\label{C_1AC_2}
\end{align}
\end{proposition}
\begin{proof}
Some easy manipulations yield these identities by using \eqref{YXY'} and \eqref{ZXZ'}.
\end{proof}

The following maps for Lie algebra valued differential forms are induced by ones of a differential 2-crossed module in the section \ref{G2}.
We define a bilinear map $\overline{\sigma}: \Lambda^{t_1}(M, \mathcal{h}) \times \Lambda^{t_2}(M, \mathcal{h}) \longrightarrow \Lambda^{t_1+t_2}(M, \mathcal{g})$ by 
\begin{align}\label{123}
\langle \overline{\sigma}(B_1,B_2), A\rangle = (-1)^{kt_2+1} \langle B_1, A\wedge^\vartriangleright B_2\rangle,
\end{align}
for $A\in \Lambda^k(M, \mathcal{g}), B_1 \in \Lambda^{t_1}(M, \mathcal{h})$ and $B_2 \in \Lambda^{t_2}(M, \mathcal{h})$,
and have 
\begin{align}\label{AB_1B_2}
\langle  A,\overline{\sigma}(B_1,B_2)\rangle = (-1)^{t_1t_2+1} \langle  A\wedge^\vartriangleright B_2,B_1\rangle.
\end{align}
Similarly, define a bilinear map $\overline{\kappa}: \Lambda^{q_1}(M, \mathcal{l}) \times \Lambda^{q_2}(M, \mathcal{l}) \longrightarrow \Lambda^{q_1+q_2}(M, \mathcal{g})$ by 
\begin{align}\label{CCA}
\langle \overline{\kappa}(C_1,C_2), A\rangle = (-1)^{kq_2+1} \langle C_1, A\wedge^\vartriangleright C_2\rangle,
\end{align} 
for $C_1 \in \Lambda^{q_1}(M, \mathcal{l})$ and $C_2 \in \Lambda^{q_2}(M, \mathcal{l})$,
and have
\begin{align}\label{AC_1C_2}
\langle  A,\overline{\kappa}(C_1,C_2)\rangle = (-1)^{q_1q_2+1} \langle  A\wedge^\vartriangleright C_2,C_1\rangle.
\end{align}
Besides, we define bilinear maps $\overline{\eta_i}: \Lambda^q (M,\mathcal{l}) \times \Lambda^t(M, \mathcal{h}) \longrightarrow \Lambda^{q+t}(M,\mathcal{h})$ $( i = 1,2)$ by 
$$\overline{\eta_i}(C, B):= \sum\limits_{a,b} C^b \wedge B^a \eta_i(Z_b, Y_a), $$
for $B= \sum\limits_{a} B^a Y_a \in \Lambda^t(M, \mathcal{h})$, $C= \sum\limits_{b} C^b Z_b \in \Lambda^q(M, \mathcal{l})$.
There are two identities
\begin{align}\label{B_1B_2C}
\langle  B_1 \wedge^{\left\{\ ,\   \right\}} B_2, C \rangle = (-1)^{t_1(t_2 + q) +1}\langle B_2, \overline{\eta_1}(C, B_1) \rangle,
= (-1)^{t_2 q +1}\langle B_1, \overline{\eta_2}(C, B_2) \rangle 
\end{align}
by using \eqref{YY'Z}.
Moreover, we have
\begin{align}
\langle B, \alpha^*(A)\rangle&= \langle\alpha(B), A\rangle,\label{BAB}\\
\langle C, \beta^*(B)\rangle&= \langle\beta(C), B\rangle,\label{ASD}
\end{align}
induced by $\langle Y, \alpha^*(X)\rangle_\mathcal{h}=\langle \alpha(Y), X\rangle_\mathcal{g}$ and $\langle Z, \beta^*(Y)\rangle_\mathcal{l}=\langle \beta(Z), Y\rangle_\mathcal{h}$.
Finally, a trilinear map $\overline{\theta}: \Lambda^{t_1}(M, \mathcal{h}) \times \Lambda^{t_2}(M, \mathcal{h})  \times \Lambda^{q}(M, \mathcal{l}) \longrightarrow \Lambda^{t_1+t_2 + q}(M, \mathcal{g}) $ is defined by
\begin{align}\label{321}
\langle \overline{\theta}(B_1, B_2, C) , A \rangle = (-)^{(q+k)(t_1 + t_2) +1 } \langle C, (A \wedge^{\vartriangleright} B_1) \wedge^{\left\{ \ ,\  \right\}}B_2\rangle.
\end{align}

\section{ YM as  a BFYM theory }\label{YMBFYM}
The BF theory is a topological field theory, which  has less physical degrees of freedom than the YM theory and doesn't depend on the spacetime metric.
Given a Lie group $G$ and its corresponding Lie algebra $\mathcal{g}$, we consider a gauge field (i.e. a $\mathcal{g}$-valued differential one-form) $A$ over a $d$-dimensional manifold $M$ $(d \geqslant 2)$ and its corresponding field strength two-form $F = dA + A \wedge A$.  The BF action principle has the form 
\begin{align}\label{pBF}
S_{BF} = \int_{M} \langle \overline{B}, F \rangle,
\end{align}
where $\overline{B}$ is a $\mathcal{g}$-valued Lagrange multiplier
$(d-2)$-form. 
Taking the variations with respect to $\overline{B}$ and $A$ respectively, we can obtain the following equations of motion 
\begin{align}
F =0,\ \ \ \ \ \ \ \ \ d \overline{B} +A \wedge^{[ \ , \ ]} \overline{B}=0.
\end{align}
The first equation implies that $A$ is a flat connection,   and the second equation implies that the Lagrange multiplier $\overline{B}$ is a constant field.
Therefore, there are no local propagating degrees of freedom, which shows that the BF theory is a topological field theory.
For more discussions and developments of this topic, we suggest \cite{Celada:2016jdt, Freidel:1999rr, EW, GTH, MBGT, MBGT1, DMM, VCMM, CMVV, CGRS, JCB, CRR, ASS} and references therein.

In this paper, we are interested in the BFYM theory \cite{Celada:2016jdt, ASC}, which is a  nontopological deformed BF theory.
The BFYM action is given by
\begin{align}\label{ddd}
S_{BFYM} = \int_{M} \langle \overline{B}, F \rangle + e^2 \langle \overline{B}, \ast \overline{B} \rangle
\end{align}
over a $d$-dimensional curved manifold $M$ $(d \geqslant 2)$  with a fixed metric tensor $g_{\mu \nu}$ of signature $(\sigma, +1, ..., +1)$. Here $\sigma =-1  (\sigma = +1)$ corresponds to the Lorentzian (Riemannian) case, and $e$ is a nonvanishing coupling constant and  $\ast$ is the internal Hodge star operator  on $M$. See appendix \ref{star}.

Varying the action \eqref{ddd} with respect to $\overline{B}$ and $A$ respectively, we obtain the equations of motion 
\begin{align}
F + 2 e^2 \ast \overline{B} = 0, \ \ \ d\overline{B} + A\wedge^{[,]} \overline{B}= 0.
\end{align}
By the first equation, we can see that $\overline{B}$ must be equal to $ - \dfrac{\sigma}{2e^2}\ast \mathcal{F}$.
Then substitute it into the second equation, and we obtain the Yang-Mills equation, i.e. $d\ast F + A\wedge^{[,]} \ast F = 0$. Besides, substituting $\overline{B}$  into \eqref{ddd}, we obtain the Yang-Mills action
\begin{align}\label{YM}
S_{YM} = - \dfrac{\sigma}{4e^2} \int_{M} \langle F, \ast F \rangle.
\end{align}
Thus, the space of solutions of the BFYM theory is corresponding to  the space of solutions of the YM theory.

For the vanishing coupling constant, i.e. $e \rightarrow 0$, the action \eqref{ddd} reduces to the topological BF action \eqref{pBF}.
For the nonvanishing $e$, the term $e^2 \langle \overline{B}, \ast \overline{B} \rangle $ in  \eqref{ddd} violates the topological symmetry and includes the degrees of freedom of the YM theory.
Consequently, the BFYM and YM theories are classically equivalent. 
We  note that this equivalence also holds at the quantum level  \cite{FUCITO1997259, MARTELLINI199762}, but we will not develop this point here.

Under a gauge transformation with $g\in G$, the gauge field $A$ transforms as
\begin{align}
A \longrightarrow A'= g^{-1}A g + g^{-1}dg,
\end{align}
and the field strength $F$ changes as
\begin{align}
F \longrightarrow F'= g^{-1}F g.
\end{align}
It is easy to prove that the YM action \eqref{YM} is gauge invariant, i.e., is invariant under the above gauge transformation. Besides, we introduce the gauge transformation acts on the Lagrange multiplier
\begin{align}
\overline{B} \longrightarrow \overline{B}'= g^{-1}\overline{B}g.
\end{align}
Then, it implies that the BF action \eqref{pBF} and BFYM action \eqref{ddd} are both gauge invariant. 
Moreover, the correspondence
between the BFYM and YM theories is preserved by this gauge equivalence. 

\section{Higher YM as higher BFYM theories} \label{HBFYM}
\subsection{2YM as a 2BFYM theory}\label{2BFYM}
In the framework of the higher gauge theory, 
the BF theory have been generalized to the 2BF theory \cite{Martins:2010ry, FGHPEM} based on a higher group. In this section, we first review the 2BF theory and the 2YM theory, then construct a deformed 2BF action to define the dynamics of the 2YM theory.
We will refer to this type of action as a 2BFYM action.
\subsubsection{2BF theory}
Given a Lie crossed module $(H,G; \alpha,\vartriangleright)$, an associated differential crossed module $(\mathcal{h}, \mathcal{g}; \alpha, \vartriangleright)$ can be constructed.
On a $d$-dimensional manifold $M$ $(d \geqslant 3)$,
a 2-gauge field (called a 2-connection) $(A, B)$ consists of a $\mathcal{g}$-valued one-form $A$ and an $\mathcal{h}$-valued two-form $B$. The fake curvature two-form and a 2-curvature three-form can be given by
\begin{align}\label{cur}
\mathcal{F} = dA + A \wedge A - \alpha(B),\ \ \mathcal{G} = dB + A \wedge^{\vartriangleright} B.
\end{align}
A 2-gauge field $(A, B)$ will be called fake-flat, if the fake curvature vanishes, i.e. $\mathcal{F}=0$. Similarly, $(A, B)$ will be called 2-flat, if the 2-curvature vanishes, i.e. $\mathcal{G}=0$. Besides, $(A, B)$ will be called flat, if $(A, B)$ is not only fake flat, but also 2-flat, i.e. $\mathcal{F}=\mathcal{G}=0$.

The corresponding 2BF action can be constructed by  using the two curvatures \eqref{cur} as
\begin{align}\label{2BF}
S_{2BF} = \int_{M} \langle \overline{B}, \mathcal{F} \rangle + \langle \overline{C}, \mathcal{G} \rangle,
\end{align}
where $\overline{B} \in \Lambda^{d - 2} (M, \mathcal{g})$ and  $\overline{C} \in \Lambda^{d-3} (M, \mathcal{h})$ are two Lagrange multipliers.

Take the variational derivative of $S_{2BF}$ with respect to  $\overline{B}$ and $\overline{C}$, then we can get the equations of motion
\begin{align}\label{BFEM}
\delta \overline{B} :  \mathcal{F} =0,\ \ \ \ 
\delta \overline{C} :  \mathcal{G} = 0.
\end{align}
Equations \eqref{BFEM} imply that the 2-gauge field $(A, B)$ is flat, i.e. fake flat and 2-flat. The vanishing of the fake curvature implies that the solution of the 2BF theory is a 2-connection $(A, B)$ on a trivial principal 2-bundle with a strict structure 2-group over $M$, which is equivalent to a 2-holonomy \cite{JBUS}.

Similarly, taking the variational derivative of $S_{2BF}$ with respect to $A$ and $B$, we have
\begin{align}
&	\delta A :  d \overline{B} +A \wedge^{[\ ,\ ]} \overline{B} = (-1)^{d+1} \overline{\sigma} (\overline{C}, B),\\
&	\delta B : d \overline{C} +A \wedge^{\vartriangleright} \overline{C} = (-1)^d \alpha^* (\overline{B}).
\end{align}
The two equations have only trivial solutions, or one can use the Hamiltonian analysis to show that there are no local propagating degrees of freedom \cite{AMMA, AMMA1}. In other wards, it shows the topological nature of the 2BF theory. 

The manipulation of the variation is quite similar to that  given in section \ref{2bfym}, thus we will not give this proof here.

\subsubsection{2YM theory}\label{2YMM}
In 2002, Baez has generalized YM theory to a kind of 2YM theory on a 4-dimensional manifold $M$ within the framework of the 2-gauge theory \cite{2002hep.th....6130B}. In this section, we extend these results on a $d$-dimensional manifolds $M$ $(d \geq 3)$ with a metric tensor $g_{\mu\nu}$ of signature $(\sigma, +1, \cdots, +1)$.
For a 2-gauge fields $(A, B)$ and the corresponding curvature form $(\mathcal{F}, \mathcal{G})$, 
we copy the usual formula for the YM action \eqref{YM} and write down the 2YM action
\begin{align}\label{2ym}
S_{2YM} = - \dfrac{\sigma}{4e^2} \int_{M}  (\langle \mathcal{F} , \ast \mathcal{F} \rangle +\langle \mathcal{G} , \ast \mathcal{G} \rangle).
\end{align}

Take the variational derivative of $S_{2YM}$ with respect to the 2-gauge field $(A, B)$ and the Lagrange multipliers $\overline{B}$ and $\overline{C}$, then obtain 
\begin{align}
\delta S_{2YM} = - \dfrac{\sigma}{2e^2} \int_{M} (\langle \delta \mathcal{F} , \ast \mathcal{F} \rangle +\langle \delta\mathcal{G}, \ast \mathcal{G} \rangle ).
\end{align}
For the first term, we have
\begin{align*}
\langle \delta \mathcal{F},\ast \mathcal{F}\rangle &= \langle \delta(dA + A \wedge A - \alpha(B)), \ast \mathcal{F} \rangle\\
&= \langle \delta A, d\ast \mathcal{F}+ A\wedge^{[,]}\ast \mathcal{F} 	\rangle-\langle \delta B, \alpha^*(\ast \mathcal{F})\rangle,
\end{align*}
by using \eqref{A_1A_2A_3} and \eqref{BAB}.
Similarly, for the second term, we have   
\begin{align*}
\langle \delta\mathcal{G}, \ast \mathcal{G}\rangle &= \langle \delta(dB + A \wedge^{\vartriangleright} B), \ast \mathcal{G} \rangle\\
&= - \langle \delta A, \overline{\sigma}(\ast \mathcal{G}, B) \rangle -\langle \delta B, d\ast \mathcal{G} + A\wedge^\vartriangleright \ast \mathcal{G} \rangle
\end{align*}
by using \eqref{B_1AB_2} and \eqref{AB_1B_2}.

Thus, we see that the variation of the action vanishes if and only if the following 2YM field equations hold
\begin{align}\label{2ymem}
d\ast \mathcal{F} +A\wedge^{[\ ,\ ]} \ast \mathcal{F}=\overline{\sigma}(\ast \mathcal{G}, B),\ \ \ 
d\ast \mathcal{G} + A\wedge^\vartriangleright \ast \mathcal{G}=-\alpha^*(\ast \mathcal{F}).
\end{align}
In the above calculation, we assume the variations $\delta A$ and $\delta B$ are compactly supported. Thus, some boundary terms are ignored in this case.

Consider a Euclidean 2-group $(\mathbb{R}^n, SO(n); \alpha, \vartriangleright)$ or Poincar$\acute{e}$ 2-group $(\mathbb{R}^n, SO(1, n-1); \alpha, \vartriangleright)$, where $\vartriangleright$ is the representation of $G$ on $H$, and $\alpha$ is a trivial map. We have the associated differential crossed module $(\mathbb{R}^n,so(n); \alpha, \vartriangleright)$ or $(\mathbb{R}^n, so(1, n); \alpha, \vartriangleright)$ with a null map $\alpha$.
This implies in particular that the fake curvature $\mathcal{F}= dA + A\wedge A$. If the 2-gauge field $(A, B)$ is fake flat, i.e. $\mathcal{F}=0$, the flatness of $A$ implies $A=0$. Thus, the 2-curvature three-form $\mathcal{G}= dB$. Therefore the 2YM action \eqref{2ym} agrees with that of Abelian 2-form electrodynamics \cite{HP}.
A similar result can be shown for all tangent 2-groups, in which $H=\mathcal{g}$ is a Lie algebra of any Lie group $G$, with the adjoint representation
$\vartriangleright$ and the trivial map $\alpha$.

\subsubsection{2BFYM theory}\label{2bfym}
Based on the above 2-categorical generalization of BF and YM theories,  a similar argument will be applied to the BFYM theory.  This is the first major theme of the present paper.

We copy the usual formal for the BFYM action \eqref{ddd} and construct a new action as
\begin{align}\label{ase}
S_{2BFYM} = \int_{M} \langle \overline{B}, \mathcal{F} \rangle +  e^2 \langle \overline{B}, \ast \overline{B} \rangle + \langle \overline{C}, \mathcal{G} \rangle + e^2 \langle \overline{C}, \ast \overline{C} \rangle.
\end{align}
We will refer to this action as the 2BFYM action. 

Take the variational derivative of  $S_{2BFYM} $ with respect to $A$ and $B$, then obtain
\begin{align}\label{0}
\delta S_{2BFYM} = \int_{M} \delta \langle \overline{B}, \mathcal{F} \rangle + 2e^2 \langle \delta \overline{B}, \ast \overline{B} \rangle + \delta \langle \overline{C}, \mathcal{G} \rangle + 2e^2\langle \delta \overline{C}, \ast \overline{C}\rangle.
\end{align}
For the first term, we have
\begin{align}\label{1}
\delta \langle \overline{B}, \mathcal{F} \rangle= \langle \delta \overline{B}, \mathcal{F} \rangle + \langle \delta A, d \overline{B} + A \wedge ^{[,]} \overline{B} \rangle - \langle \delta B, \alpha^{*}(\overline{B}) \rangle,
\end{align}
by using \eqref{A_1A_2A_3} and \eqref{BAB}.
And for the third term, we have
\begin{align}\label{2}
\delta \langle \overline{C}, \mathcal{G} \rangle= \langle \delta \overline{C}, \mathcal{G} \rangle + (-1)^d \langle \delta B, d \overline{C} + A \wedge^{\vartriangleright}\overline{C}\rangle +(-1)^d \langle \delta A, \overline{\sigma}(\overline{C}, B)\rangle,
\end{align}
by using \eqref{B_1AB_2} and \eqref{AB_1B_2}. 
Substituting \eqref{1} and \eqref{2} into \eqref{0}, obtain
\begin{align}
\delta S_{2BFYM} &= \int_{M} \langle \delta \overline{B}, \mathcal{F} + 2e^2 \ast \overline{B} \rangle + \langle \delta A, d \overline{B} + A \wedge^{[,]} \overline{B} + (-1)^d \overline{\sigma}(\overline{C}, B) \rangle \nonumber \\
& \ \ \ + \langle \delta \overline{C}, \mathcal{G} + 2e^2 \ast \overline{C} \rangle + (-1)^d \langle \delta B, d \overline{C} + A \wedge^{\vartriangleright} \overline{C} - (-1)^d \alpha^*(\overline{B}) \rangle.
\end{align}
It is easy to see that the equations of motion of the 2BFYM theory are
\begin{align}
&\delta \overline{B} : \ \ \mathcal{F} + 2 e^2 \ast \overline{B} = 0,\label{dd1}\\
&	\delta \overline{C} : \ \ \mathcal{G} +  2 e^2 \ast \overline{C} = 0,\label{dd2}\\
&	\delta A : \ \ d\overline{B} + A\wedge^{[\ ,\ ]} \overline{B}= (-1)^{d+1} \overline{\sigma} (\overline{C}, B),\label{d1}\\
&	\delta B : \ \ d\overline{C} + A\wedge^{\vartriangleright} \overline{C} = (-1)^d \alpha^* (\overline{B}).\label{d2}
\end{align}
From the first two equations, we get $\overline{B}= - \dfrac{\sigma}{2e^2}\ast \mathcal{F}$, and $\overline{C}= (-1)^d \dfrac{\sigma}{2e^2}\ast \mathcal{G}$ by using \eqref{ast}. Substituting them into the last two equations, we can obtain the 2-form Yang-Mills equations \eqref{2ymem}.
On the other side, substitute the solution $\overline{B}$
and $\overline{C}$ back into \eqref{ase}, then we obtain the 2YM action \eqref{2ym}.
Thus, the space of solutions of the 2BFYM theory is corresponding to that of the 2YM theory. In the limit of vanishing coupling, i.e. $e \rightarrow 0$, the action \eqref{ase} reduces to the topological 2BF action \eqref{2BF}.
Consequently, the 2BFYM theory is classically equivalent to the 2YM theory. 

\emph{The gauge symmetry of 2BFYM theory.}
We will be only interested in the local aspects of 2-connections, and  consider the following two types of
gauge transformations \cite{Martins:2010ry}.
\begin{itemize}
\item [1)] Thin: $A' = g^{-1} A g + g^{-1} dg$,\ \ \ $B' = g^{-1} \vartriangleright B$,\ \ \ $\overline{B}' = g^{-1} \overline{B} g$,\ \ \ $\overline{C}'= g^{-1} \vartriangleright \overline{C}$;
\item [2)] Fat: $A' = A + \alpha(\phi)$,\ \ \ $B' = B + d \phi + A \wedge^{\vartriangleright} \phi + \phi \wedge \phi,$\ \ \ $\overline{B}' = \overline{B} + \overline{\sigma}(\overline{C}, \phi)$, \ \ \ $\overline{C}' = \overline{C},$
\end{itemize}
where $g \in G$ and $\phi \in \Lambda^1(M, \mathcal{h})$.

Under the thin gauge transformation, the corresponding curvatures change as
\begin{align}
\mathcal{F}' = g^{-1} \mathcal{F} g,\ \ \ \mathcal{G}' = g^{-1} \vartriangleright \mathcal{G}.
\end{align}
Due to the $G$-invariance  of the bilinear forms, have
\begin{align}
\langle g^{-1} \overline{B} g, g^{-1} \mathcal{F} g \rangle = \langle \overline{B}, \mathcal{F} \rangle,\ \ \ \langle g^{-1}\vartriangleright \overline{C}, g^{-1} \vartriangleright \mathcal{G} \rangle = \langle \overline{C}, \mathcal{G} \rangle,
\end{align}
then the 2BF action \eqref{2BF} is invariant, i.e. $S_{2BF}'= S_{2BF}$. 
Similarly, have
\begin{align}
\langle g^{-1} \overline{B} g, g^{-1} \ast \overline{B} g \rangle = \langle \overline{B}, \ast \overline{B} \rangle,\ \ \ \langle g^{-1}\vartriangleright \overline{C}, g^{-1} \vartriangleright \ast \overline{C} \rangle = \langle \overline{C}, \ast \overline{C} \rangle,
\end{align}
then we have $S_{2BFYM}' = S_{2BFYM}$.
Namely, the 2BFYM action is thin gauge invariant.

Under the fat gauge transformation, the curvatures change as 
\begin{align}
\mathcal{F} ' =  \mathcal{F},\ \ \ 
\mathcal{G} ' =  \mathcal{G} + \mathcal{F} \wedge^{\vartriangleright} \phi.
\end{align}
Due to
$$\langle \overline{\sigma}(\overline{C}, \phi), \mathcal{F} \rangle = - \langle \overline{C}, \mathcal{F} \wedge^{\vartriangleright} \phi \rangle, $$
by using \eqref{123},  the 2BF action \eqref{2BF} is fat invariant.
Besides for 2BFYM action, have
\begin{align}\label{fat2BF}
S_{2BFYM}' = S_{2BFYM} + e^2 (\langle \overline{\sigma}(\overline{C} , \phi),  \ast \overline{\sigma}(\overline{C} , \phi) + 2 \ast \overline{B} \rangle,
\end{align}
then it implies that the 2BFYM action  is not fat gauge invariant. 
However, it follows from \eqref{fat2BF} that the 2BFYM action  is fat gauge invariant with the constrained condition
$$\overline{\sigma}(\overline{C} , \phi)=0, \ \ \ or \ \ \ 	\overline{\sigma}(\overline{C} , \phi) =- 2 \overline{B}.$$

\subsection{3YM as a 3BFYM theory}\label{3BFYM}
In this section, we first introduce 3BF \cite{TRMV, Radenkovic:2019qme} and 3YM \cite{sdh} theories based the categorical generalization. Then we  construct a deformed 3BF action to define the dynamics of the 3YM theory. Be similar to 2BFYM action, we will refer to this type of action as a 3BFYM action.
\subsubsection{3BF theory} 
Given a 2-crossed module $(L, H, G; \beta, \alpha, \vartriangleright, \{ \ , \ \})$, an associated differential 2-crossed module $(\mathcal{l}, \mathcal{h}, \mathcal{g}; \beta, \alpha, \vartriangleright, \{ \ , \ \})$ can be constructed. We consider a 3-gauge field (called a 3-connection) $(A, B, C)$ on a $d$-dimensional manifold $M$ $(d \geqslant 4)$, consisting of 
$$A \in \Lambda^1(M, \mathcal{g}),\ \ \ \ B \in \Lambda^2(M, \mathcal{h}),\ \ \ \ C \in \Lambda^3(M, \mathcal{l}).$$ 
The related fake curvature 2-form, fake 2-curvature 3-form and 3-curvature 4-form are respectively, given by:
\begin{align}
\mathcal{F} = dA + A \wedge A &- \alpha(B),\ \ \ \ \ \ 
\mathcal{G} = dB + A\wedge^{\vartriangleright} B - \beta(C),\nonumber\\
\mathcal{H} =& dC + A \wedge^\vartriangleright C + B \wedge^{\left\{ \ ,\  \right\}} B.
\end{align}
A 3-gauge field $(A, B, C)$ will be called fake 1-flat, if the fake curvature vanishes, i.e. $\mathcal{F}=0$, and fake flat, if it is fake 1-flat and the fake 2-curvature vanishes, i.e. $\mathcal{G}=0$. Furthermore, $(A, B, C)$ will be called flat, if it is fake flat and the 3-curvature vanishes, i.e. $\mathcal{H}=0$.

Continuing the categorical generalization one step further, the 2BF action has been generalized to the 3BF action \cite{ TRMV11, RV}.
The corresponding 3BF action can be written as
\begin{align}\label{3BF}
S_{3BF} = \int_{M} \langle \overline{B}, \mathcal{F} \rangle + \langle \overline{C}, \mathcal{G} \rangle + \langle \overline{D}, \mathcal{H} \rangle,
\end{align}
where $\overline{B} \in \Lambda^{d-2} (M, \mathcal{g})$, $\overline{C} \in \Lambda^{d-3} (M, \mathcal{h})$ and $\overline{D} \in \Lambda^{d-4} (M, \mathcal{l})$ are Lagrange multipliers. 

Varying the action $S_{3BF}$ with respect to $\overline{B}, \overline{C}$ and $\overline{D}$, we obtain the equations of motion,
\begin{align}\label{eq3}
\delta \overline{B} :\mathcal{F} =0,\ \ \
\delta \overline{C}:\mathcal{G} =0,\ \ \
\delta \overline{D}:\mathcal{H} = 0.
\end{align}
Equations \eqref{eq3} imply that the 3-gauge field $(A, B, C)$ is flat. The vanishing of these fake curvatures implies that the solution $(A, B, C)$ of 3BF theory is a 3-connection on a trivial principal 3-bundle with a strict structure 3-group over $M$, which is equivalent to a 3-holonomy \cite{Saemann:2013pca}.

Similarly, varying with respect to
the 3-gauge field $(A, B, C)$,  the equations of motion are given by
\begin{align*}
&\delta A:\ \ \ \ d \overline{B} +A \wedge^{[\ ,\ ]} \overline{B} = (-1)^{d+1} \overline{\sigma} (\overline{C}, B) + (-1)^d \overline{\kappa} (\overline{D}, C),\\
&\delta B:\ \ \ \ d \overline{C} +A \wedge^{\vartriangleright} \overline{C} = (-1)^d \alpha^* (\overline{B}) +(-1)^{d} (\overline{\eta_2}(\overline{D}, B) + \overline{\eta_1}(\overline{D}, B)),\\
&\delta C:\ \ \ \ d \overline{D} + A\wedge^\vartriangleright \overline{D} = (-1)^{d+1} \beta^* (\overline{C}).
\end{align*}
And the three equations have only trivial solutions showing that there are no local propagating degrees of freedom. That is to say the 3BF theory is topological.
A detail manipulation of the variation will be given in appendix \ref{em3BF}. 

\subsubsection{3YM theory}
Without loss of generality, we apply a similar generalization to the 2YM theory. 
We consider a 3-gauge field $(A, B, C)$ on a $d$-dimensional manifold $M$ $(d \geq 4)$ with a metric tensor $g_{\mu\nu}$ of signature $(\sigma, +1, \cdots, +1)$. 
Then, be similar to the construction of Yang-Mills and 2-form Yang-Mills actions, the 3-form Yang-Mills action \cite{sdh} can be defined by the corresponding curvature form $(\mathcal{F}, \mathcal{G}, \mathcal{H})$ as
\begin{align}\label{3ymac}
S_{3YM}= - \dfrac{\sigma}{4e^2} \int_{M}( \langle \mathcal{F} , \ast \mathcal{F} \rangle +\langle \mathcal{G} , \ast \mathcal{G} \rangle+ \langle \mathcal{H} , \ast \mathcal{H} \rangle).
\end{align}

Varying the action with respect to $A$, $B$ and $C$, we obtain 
\begin{align*}
\delta S_{3YM} = &- \dfrac{\sigma}{2e^2} \int_{M}  (\langle \delta A, d\ast \mathcal{F}+ A\wedge^{[,]}\ast \mathcal{F} + (-1)^{d-1} \overline{\kappa}(\ast\mathcal{H}, C)-\overline{\sigma}(\ast \mathcal{G}, B)\rangle
\\[2mm]
&\ \ - \langle \delta B,d \ast \mathcal{G} + A\wedge^\vartriangleright \ast \mathcal{G}+ \overline{\eta_2}(\ast \mathcal{H}, B) + \overline{\eta_1}(\ast \mathcal{H}, B)+\alpha^*(\ast \mathcal{F})\rangle\\[2mm]
&\ \ + \langle \delta C, d\ast \mathcal{H}+A\wedge^\vartriangleright \ast\mathcal{H} - \beta^*(\ast \mathcal{G})\rangle).
\end{align*}
A detail manipulation of the variation will be given in appendix \ref{v3bf}.
We see that the variation of the action vanishes for $\delta A $, $\delta B$ and $\delta C$ if and only if the following field equations hold
\begin{align}\label{3ymeq}
&d\ast \mathcal{F}+ A\wedge^{[\ ,\ ]}\ast \mathcal{F} =\overline{\sigma}(\ast \mathcal{G}, B)+ (-1)^{d} \overline{\kappa}(\ast\mathcal{H}, C),\nonumber\\
&d \ast \mathcal{G} + A\wedge^\vartriangleright \ast \mathcal{G}=-\overline{\eta_2}(\ast \mathcal{H}, B) - \overline{\eta_1}(\ast \mathcal{H}, B)-\alpha^*(\ast \mathcal{F}),\nonumber\\
&d\ast \mathcal{H}+A\wedge^\vartriangleright \ast\mathcal{H} = \beta^*(\ast \mathcal{G}).
\end{align}

In the above calculation, we assume the variations $\delta A$, $\delta B$ and $\delta C$ are compactly supported. Thus, some boundary terms are ignored in this case.

\subsubsection{3BFYM theory}
In this subsection, we turn to the 3-categorical generalization of BFYM theory. This is the second major theme of the present paper. The main problem is to push the 2BFYM action further and set up a sort of 3BFYM action, which is precisely equivalent to the 3-form Yang-Mills theory.

Similarly, we copy the usual formal for the 2BFYM action \eqref{ase} and construct a new action as
\begin{align}\label{bse}
S_{3BFYM} =  \displaystyle{\int_{M} \langle \overline{B}, \mathcal{F} \rangle +  e^2 \langle \overline{B}, \ast \overline{B} \rangle + \langle \overline{C}, \mathcal{G} \rangle + e^2 \langle \overline{C}, \ast \overline{C} \rangle+ \langle \overline{D}, \mathcal{H} \rangle + e^2 \langle \overline{D}, \ast \overline{D} \rangle }.
\end{align}
We will refer to this action as the 3BFYM action.

Vary the action with respect to the 3-gauge field and the Lagrange multipliers 
\begin{align}
&\delta S_{3BFYM} \nonumber\\
= &\int_{M} \delta \langle \overline{B}, \mathcal{F} \rangle +  e^2 \delta \langle \overline{B}, \ast \overline{B} \rangle + \delta \langle \overline{C}, \mathcal{G} \rangle +  e^2 \delta \langle \overline{C}, \ast \overline{C} \rangle+ \delta \langle \overline{D}, \mathcal{H} \rangle + e^2 \delta \langle \overline{D}, \ast \overline{D} \rangle \nonumber\\
=& \int_{M} \delta \langle \overline{B}, \mathcal{F} \rangle +  2e^2  \langle \delta \overline{B}, \ast \overline{B} \rangle + \delta \langle \overline{C}, \mathcal{G} \rangle +  2e^2  \langle \delta \overline{C}, \ast \overline{C} \rangle+ \delta \langle \overline{D}, \mathcal{H} \rangle + 2e^2  \langle \delta \overline{D}, \ast \overline{D} \rangle \nonumber\\
= &\int_{M} \langle \delta \overline{B}, \mathcal{F} + 2 e^2 \ast \overline{B} \rangle + \langle \delta A, d \overline{B} + A \wedge^{[\ ,\ ]} \overline{B} + (-1)^d \overline{\sigma}(\overline{C}, B) + (-1)^{d+1} \overline{\kappa}(\overline{D}, C) \rangle \nonumber\\
&+\langle \delta \overline{C}, \mathcal{G} +2e^2 \ast \overline{C} \rangle + (-1)^d \langle \delta B, d \overline{C} + A \wedge^{\vartriangleright}\overline{C} - (-1)^d \alpha^*(\overline{B})- (-1)^{d} (\overline{\eta_2}(\overline{D}, B)  \nonumber\\
& + \overline{\eta_1}(\overline{D}, B))\rangle+ \langle \delta \overline{D}, \mathcal{H} + 2e^2 \ast \overline{D} \rangle + \langle \delta C, d \overline{D} + A\wedge^{\vartriangleright}\overline{D} + (-1)^d\beta^*(\overline{C})
\end{align}
by using \eqref{46}, \eqref{47} and \eqref{48}.

Then, we obtain the following equations of motion
\begin{align}
&\delta \overline{B} : \ \ \mathcal{F} + 2 e^2 \ast \overline{B} = 0,\ \ \ 
\delta \overline{C} : \ \ \mathcal{G} + 2 e^2 \ast \overline{C} = 0,\ \ \ 
\delta \overline{D} : \ \ \mathcal{H} + 2 e^2 \ast \overline{D} = 0,\nonumber\\
&	\delta A : \ \ d \overline{B} +A \wedge^{[\ ,\ ]} \overline{B} = (-1)^{d+1} \overline{\sigma} (\overline{C}, B) + (-1)^d \overline{\kappa} (\overline{D}, C),\nonumber\\
&	\delta B : \ \ d \overline{C} +A \wedge^{\vartriangleright} \overline{C} = (-1)^d \alpha^* (\overline{B}) +(-1)^{d} (\overline{\eta_2}(\overline{D}, B) + \overline{\eta_1}(\overline{D}, B)),\nonumber\\
&	\delta C : \ \ d \overline{D} + A\wedge^\vartriangleright \overline{D} = (-1)^{d+1} \beta^* (\overline{C}).
\end{align}
From the first three equations, we get $\overline{B}= - \dfrac{\sigma}{2e^2}\ast \mathcal{F}$, $\overline{C}= (-1)^d \dfrac{\sigma}{2e^2}\ast \mathcal{G}$ and $\overline{D}= - \dfrac{\sigma}{2e^2}\ast \mathcal{H}$ by using \eqref{ast}.
Substitute them into the last three equations, then we can obtain the 3YM equations of motion \eqref{3ymeq}.
Besides, substitute the solution $\overline{B}$, $\overline{C}$ and $\overline{D}$ back into \eqref{bse}, then we obtain the 3YM action \eqref{3ymac}.
Thus, the space of solution of the 3BFYM theory is corresponding to that of the 3YM theory.
And in the limit of vanishing coupling, $e \rightarrow 0$,  the action \eqref{bse} will reduce to the topological 3BF action.    Consequently, the 3BFYM theory is  classically equivalent to the 3YM theory.

\emph{The gauge symmetry of 3BFYM theory.}
We consider the gauge symmetries of 3BFYM theory under the following three types of gauge transformations appearing in \cite{TRMV}.  
\begin{itemize}
\item [1)] $G$-gauge transformation for $g \in G$:
\begin{align}
A' = g^{-1} A g + g^{-1} dg, \ \ \ B' = g^{-1} \vartriangleright B, \ \ \ C' = g^{-1} \vartriangleright C,\nonumber\\
\overline{B}' = g^{-1} \overline{B} g,\ \ \	\overline{C}' = g^{-1} \vartriangleright \overline{C}, \ \ \	\overline{D}' = g^{-1} \vartriangleright \overline{D}.
\end{align}
Under this gauge transformation, the curvatures change as
\begin{align}
\mathcal{F}' = g^{-1} \mathcal{F} g,\ \  \mathcal{G}' = g^{-1} \vartriangleright \mathcal{G},\ \ \mathcal{H}' = g^{-1} \vartriangleright \mathcal{H}.
\end{align}
Using the $G$-invariance of the bilinear forms, have
\begin{align}
\langle g^{-1} \overline{B} g, g^{-1} \mathcal{F} g \rangle &= \langle \overline{B}, \mathcal{F} \rangle,\  \langle g^{-1}\vartriangleright \overline{C}, g^{-1} \vartriangleright \mathcal{G} \rangle = \langle \overline{C}, \mathcal{G} \rangle,\nonumber\\
\langle g^{-1} &\vartriangleright  \overline{D}, g^{-1} \vartriangleright \mathcal{H} \rangle = \langle \overline{D}, \mathcal{H} \rangle,
\end{align}
then the 3BF action is $G$-gauge invariant, i.e., $S_{3BF}' = S_{3BF}.$ Besides, have
\begin{align}
\langle g^{-1} \overline{B} g, g^{-1} \ast \overline{B} g \rangle& = \langle \overline{B}, \ast \overline{B} \rangle,\ 
\langle g^{-1}\vartriangleright \overline{C}, g^{-1} \vartriangleright \ast \overline{C} \rangle = \langle \overline{C}, \ast \overline{C} \rangle,\nonumber\\
&	\langle g^{-1}\vartriangleright \overline{D}, g^{-1} \vartriangleright \ast \overline{D} \rangle = \langle \overline{D}, \ast \overline{D} \rangle,
\end{align}
then, we have $S_{3BFYM}' =S_{3BFYM}$. Thus, the 3BFYM action is $G$-gauge invariant under the $G$-gauge transformation.
\item [2)] $H$-gauge transformation for $\phi \in \Lambda^1(M, \mathcal{h})$:
\begin{align}
&	A' = A + \alpha(\phi), \ \ \ B' = B + d \phi + A'\wedge^{\vartriangleright}\phi - \phi \wedge \phi,\nonumber\\ \ \ \ &C' = C - B' \wedge^{\left\{\ ,\   \right\}} \phi - \phi \wedge^{\left\{\ ,\   \right\}}  B,\nonumber\\
&	 \overline{B}' = \overline{B} + \overline{\sigma}(\overline{C}', \phi) - \overline{\theta}(\phi, \phi, \overline{D}),\ \ \	\overline{C}' = \overline{C} + \eta_1(\overline{D}, \phi) + \eta_2(\overline{D}, \phi), \nonumber\\ \ \ \	& \overline{D}' = \overline{D}.
\end{align}
Under this gauge transformation, the curvatures change as
\begin{align}
\mathcal{F}' =  \mathcal{F},\ \ \mathcal{G}' = \mathcal{G} + \mathcal{F} \wedge^{\vartriangleright}\phi, \ \  \mathcal{H}' = \mathcal{H} - \mathcal{G}' \wedge^{\left\{ \ ,\  \right\}}\phi + \phi\wedge^{\left\{ \ ,\  \right\}}\mathcal{G}.
\end{align}
Then we have 
\begin{align}\label{BFt}
\langle \overline{B}', \mathcal{F}' \rangle &=
\langle  \overline{B} + \overline{\sigma}(\overline{C}', \phi) - \overline{\theta}(\phi, \phi, \overline{D}), \mathcal{F}\rangle \nonumber\\
&= \langle \overline{B}, \mathcal{F} \rangle - \langle \overline{C}', \mathcal{F} \wedge^{\vartriangleright}\phi \rangle + \langle \overline{D}, (\mathcal{F} \wedge^{\vartriangleright}\phi) \wedge^{\{ \ , \ \}}\phi \rangle
\end{align}
by using \eqref{123} and \eqref{321}, and
\begin{align}\label{CGt}
\langle \overline{C}', \mathcal{G}'\rangle&=
\langle  \overline{C} + \eta_1(\overline{D}, \phi) + \eta_2(\overline{D}, \phi), \mathcal{G} + \mathcal{F} \wedge^{\vartriangleright}\phi \rangle \nonumber\\
&= \langle \overline{C}', \mathcal{F}\wedge^{\vartriangleright}\phi \rangle + \langle \overline{C}, \mathcal{G} \rangle -  \langle \overline{D}, \phi \wedge^{\{ \ , \ \}}\mathcal{G} \rangle + \langle \overline{D}, \mathcal{G} \wedge^{\{ \ , \ \}} \phi \rangle
\end{align}
by using \eqref{B_1B_2C}, and
\begin{align}\label{DHt}
\langle \overline{D}', \mathcal{H}'\rangle &=\langle \overline{D}, \mathcal{H} - \mathcal{G}' \wedge^{\left\{ \ ,\  \right\}}\phi + \phi\wedge^{\left\{ \ ,\  \right\}}\mathcal{G} \rangle\nonumber\\
&=\langle \overline{D}, \mathcal{H} - (\mathcal{G} + \mathcal{F} \wedge^{\vartriangleright}\phi) \wedge^{\left\{ \ ,\  \right\}}\phi + \phi\wedge^{\left\{ \ ,\  \right\}}\mathcal{G} \rangle.
\end{align}
Combining \eqref{BFt}, \eqref{CGt} with \eqref{DHt}, we get 
\begin{align}
S_{3BF}' = S_{3BF},
\end{align}
i.e. the 3BF action is $H$-gauge invariant.
Since
\begin{align*}
\langle \overline{B}', \ast \overline{B}' \rangle &
= \langle \overline{B}, \ast \overline{B} \rangle + e^2 \langle \overline{\sigma}(\overline{C}', \phi)- \overline{\theta}(\phi, \phi, \overline{D}), \ast 2 \overline{B} + \ast \overline{\sigma}(\overline{C}', \phi)- \ast \overline{\theta}(\phi, \phi, \overline{D})\rangle,\\
\langle \overline{C}', \ast \overline{C}' \rangle &= \langle \overline{C}, \ast \overline{C} \rangle + e^2\langle \eta_1(\overline{D}, \phi)+ \eta_2(\overline{D}, \phi), \ast2 \overline{C} + \ast \eta_1(\overline{D}, \phi)+ \ast \eta_2(\overline{D}, \phi),
\end{align*}
we have 
$$S_{3BFYM}' =S_{3BFYM},$$
i.e. the 3BFYM action  is $H$-gauge invariant with the following constrained conditions:
\begin{equation}
\begin{cases}
\overline{\sigma}(\overline{C}', \phi)= \overline{\theta}(\phi, \phi, \overline{D})\\
or\\
2\overline{B} = \overline{\theta}(\phi, \phi, \overline{D})- \overline{\sigma}(\overline{C}', \phi)
\end{cases}
and \ \ \ \ \
\begin{cases}
\eta_1(\overline{D}, \phi)+ \eta_2(\overline{D}, \phi)=0\\
or\\
2\overline{C} = - \eta_1(\overline{D}, \phi)- \eta_2(\overline{D}, \phi).
\end{cases}
\end{equation}
\item [3)] $L$-gauge transformation: 
$$A' = A, \ \ \ B' = B + \beta(\psi), \ \ \ C' = C + d \psi + A \wedge^{\vartriangleright} \psi,$$
$$ \overline{B}' = \overline{B} - \overline{\kappa}(\overline{D}, \psi),\ \ \	\overline{C}' = \overline{C}, \ \ \	\overline{D}' =  \overline{D}.$$

Under the $L$-gauge transformation, the curvatures change as
\begin{align}
\mathcal{F}' =  \mathcal{F},\ \ 
\mathcal{G}' = \mathcal{G},\ \ 
\mathcal{H}' = \mathcal{H} - \mathcal{F}\wedge^{\vartriangleright}\psi.
\end{align}
The 3BF action is gauge invariant as
$$S_{3BF}' = S_{3BF},$$
by using \eqref{CCA}.
Due to
\begin{align}
\langle \overline{B}', \ast \overline{B}' \rangle&= \langle \overline{B}, \ast \overline{B} \rangle -\langle \overline{\kappa}(\overline{D},\psi),\ast 2 \overline{B} - \ast \overline{\kappa}(\overline{D},\psi)\rangle,\nonumber\\
\langle \overline{C}', \ast \overline{C}' \rangle &= \langle \overline{C}, \ast \overline{C} \rangle, \ \ \ \langle \overline{D}', \ast \overline{D}' \rangle = \langle \overline{D}, \ast \overline{D} \rangle, 
\end{align}
we have 
$$S_{3BFYM}' =S_{3BFYM},$$
i.e. the 3BFYM action  is  $L$-gauge invariant with a 
constrained condition
$$\overline{\kappa}(\overline{D},\psi)=0, \ \ \ \ or \ \ \ \ \overline{B}= \overline{\kappa}(\overline{D},\psi).$$
\end{itemize}

\section{Conclusion and Outlook}
In this article, we have reviewed the generalizations of the BF theory to the 2BF and 3BF theories, and  the YM theory to the 2YM and 3YM theories. 
Similarly, we generalized the BFYM theory to higher counterparts by using the categorical ladder.
One main result in this paper is that we are able to construct the 2BFYM action based on the Lie crossed module. We also proved that the 2BFYM theory is classically equivalent to the 2YM theory, and under certain constraints, this new action is gauge invariant under thin and fat gauge transformations.
Another main result is that we developed the 3BFYM theory, which can define the dynamics of the 3YM theory and is gauge invariant under the higher gauge transformations in some cases.

The equivalence between higher YM and BFYM theories may be also preserved at the quantum level.
We leave these arguments for future work.

\section*{Acknowledgment}
This work is supported by the National Natural Science Foundation of China (Nos.11871350, NSFC no. 11971322).

\begin{appendix}
\section{Internal Hodge star}\label{star}
The Hodge dual in this paper satisfies the following conventions. Given a $d$-dimensional vector space $V$  with fixed metric tensor $g_{\mu \nu}$ of signature $(\sigma, +1, ..., +1)$, where $\sigma =-1(\sigma=+1)$ for the Lorentzian(Riemannian) case.

Supposing $\{\xi \vert i=1, ..., d\}$ be an order orthonormal basis for $V$, we consider the basis of $\Lambda^r(V)$ given by elementary wedge products $\xi_I=\xi_{i_1} \wedge ... \wedge \xi_{i_r}$ for $I= {i_1, ..., i_r}$ a strictly increasing sequence of $r$ integers between $1$ and $d$. Likewise, $\Lambda^{d-r}(V)$ has a basis given by the elementary wedge products $\xi_{I'}$ for strictly increasing sequences $I'$ consisting of $d-r$ integers between $1 $ and $d$. For any two $I$ and $I'$ as above, $\xi_I \wedge \xi_{I'}\in \Lambda^d(V)$ vanishes if $I$ and $I'$ are not complementary, whereas $\xi_I \wedge \xi_{I'}= \pm \xi_1 \wedge ... \wedge \xi_d $ in the complementary case.

In our situation, we normalize the Hodge star operator 
\begin{align}
\ast: \Lambda^r(V) \longrightarrow \Lambda^{d-r}(V)
\end{align}
so that
\begin{align}
\ast (\xi_1 \wedge \cdots \wedge \xi_r)= \xi_{r+1} \wedge \cdots \wedge \xi_d,
\end{align}
for $r$ from $0$ to $d$. Writing an arbitrary element $\omega \in \Lambda^r(V)$ as
\begin{align}
\omega = \frac{1}{r!}\omega_{i_1 \cdots i_r}\xi^{i_1}\wedge \cdots \wedge \xi^{i_r},
\end{align}
this implies the Hodge dual $\ast \omega \in \Lambda^{d-r}(V)$ is given by
\begin{align}
\ast \omega = \dfrac{1}{(d-r)! r!} \epsilon^{i_1 \cdots i_r}_{\ \ \ \ \ \  j_1 \cdots j_{d-r}} \omega_{i_1 \cdots i_r} \xi^{j_1}\wedge \cdots \wedge \xi^{j_{d-r}}.
\end{align}
The Hodge star acting on $r$-forms satisfies
\begin{align}\label{ast}
\ast ^2=(-1)^{r(d-r)}\sigma.
\end{align}

\section{The variation of the 3BF action}\label{em3BF}
Take the variational derivative of $S_{3BF}$\eqref{3BF} and have
\begin{align}\label{46}
\delta \langle \overline{B}, \mathcal{F} \rangle &= \langle \delta \overline{B}, \mathcal{F}\rangle + \langle \overline{B}, \delta \mathcal{F} \rangle \nonumber\\ 
&=\langle \delta \overline{B}, \mathcal{F} \rangle + \langle \delta A, d \overline{B} + A \wedge^{[,]} \overline{B} \rangle - \langle \delta B, \alpha^* (\overline{B}) \rangle, 
\end{align}
by using the result  \eqref{A_1A_2A_3} and \eqref{BAB},and 
\begin{align}\label{47}
\delta \langle \overline{C}, \mathcal{G} \rangle =& \langle \delta \overline{C}, \mathcal{G} \rangle + \langle \overline{C}, \delta \mathcal{G} \rangle \nonumber\\ 
= &\langle \delta \overline{C}, \mathcal{G} \rangle + \langle \overline{C}, \delta(d B + A \wedge^{\vartriangleright}\overline{B}) \rangle - \langle \overline{C}, \beta(\delta C) \rangle \nonumber\\
= &\langle \delta \overline{C}, \mathcal{G} \rangle + (-1)^d \langle \delta B, d \overline{C} + A \wedge^{\vartriangleright} \overline{C} \rangle \nonumber\\
&+ (-1)^d \langle \delta A, \overline{\sigma}(\overline{C}, B)\rangle + (-1)^d \langle \delta C, \beta^*(\overline{C})\rangle,
\end{align}
by using the result  \eqref{B_1AB_2}, \eqref{AB_1B_2} and \eqref{ASD}, and
\begin{align}\label{48}
\delta \langle \overline{D}, \mathcal{H} \rangle =&  \langle \delta \overline{D}, \mathcal{H} \rangle + \langle \overline{D}, \delta \mathcal{H} \rangle \nonumber\\ 
= &\langle \delta \overline{D}, \mathcal{H} \rangle + \langle \overline{D}, d \delta C \rangle + \langle \overline{D}, \delta A \wedge^{\vartriangleright} C \rangle + \langle \overline{D}, A \wedge^{\vartriangleright}\delta C \rangle \nonumber\\
&+ \langle \overline{D}, \delta B \wedge^{\left\{ \ ,\  \right\}}B, \rangle + \langle \overline{D}, B \wedge^{\left\{ \ ,\  \right\}}\delta B\rangle \nonumber\\
= & \langle \delta \overline{D}, \mathcal{H} \rangle + \langle d \delta C, \overline{D} \rangle + \langle \delta A \wedge^{\vartriangleright}C, \overline{D} \rangle + \langle A \wedge^{\vartriangleright} \delta C, \overline{D} \rangle \nonumber\\
&+ \langle \delta B\wedge^{\left\{ \ ,\  \right\}}B, \overline{D} \rangle + \langle B \wedge^{\left\{ \ ,\  \right\}}\delta B, \overline{D} \rangle \nonumber\\
= &\langle \delta \overline{D}, \mathcal{H} \rangle + \langle \delta C, d \overline{D} \rangle +(-1)^{d+1}\langle \delta A, \overline{\kappa}(\overline{D}, C) \rangle + \langle \delta C, A \wedge^{\vartriangleright}\overline{D} \rangle \nonumber\\
&- \langle \delta B, \overline{\eta_2}(\overline{D}, B) + \overline{\eta_1}(\overline{D}, B) \rangle \nonumber\\
= &\langle \delta \overline{D}, \mathcal{H} \rangle + \langle \delta C, d \overline{D} +A \wedge^{\vartriangleright}\overline{D} \rangle +(-1)^{d+1}\langle \delta A, \overline{\kappa}(\overline{D}, C) \rangle \nonumber\\
&- \langle \delta B, \overline{\eta_2}(\overline{D}, B) + \overline{\eta_1}(\overline{D}, B) \rangle,
\end{align}
by using \eqref{C_1AC_2},  \eqref{B_1B_2C} and \eqref{AC_1C_2}.

\section{The variation of the 3YM action}\label{v3bf}

\begin{align*}
\delta S_{3YM} = - \dfrac{1}{2e^2} \int_{M}  \langle \delta \mathcal{F} , \ast \mathcal{F}\rangle +\langle \delta \mathcal{G}, \ast \mathcal{G}\rangle 	+\langle \delta\mathcal{H}, \ast \mathcal{H}\rangle =0.\\
\end{align*}

For the first term, we have
\begin{align*}
\langle \delta \mathcal{F},\ast \mathcal{F}\rangle &= \langle \delta(dA + A \wedge A - \alpha(B)), \ast \mathcal{F} \rangle\\
&= \langle \delta A, d\ast \mathcal{F}+ A\wedge^{[,]}\ast \mathcal{F} 	\rangle-\langle \delta B, \alpha^*(\ast \mathcal{F})\rangle,
\end{align*}
by using $\delta(A\wedge A) = A \wedge^{[,]} \delta A$, \eqref{A_1A_2A_3} and \eqref{BAB}.\\ 
For the second term, we have
\begin{align*}
\langle \delta \mathcal{G}, \ast \mathcal{G} \rangle &= \langle \delta(dB + A\wedge^\vartriangleright B 	- \beta(C)), \ast  \mathcal{G} \rangle\\[2mm]
&=\langle d\delta B, \ast \mathcal{G} \rangle + \langle \delta A \wedge^\vartriangleright B, 	\ast \mathcal{G}\rangle+ \langle A\wedge^\vartriangleright \delta B, \ast \mathcal{G} \rangle - \langle \beta(\delta C), \ast \mathcal{G} \rangle\\[2mm]
&= -\langle \delta B, d \ast \mathcal{G} \rangle- \langle \delta A,\overline{\sigma}(\ast \mathcal{G}, 	B)\rangle - \langle \delta B, A\wedge^\vartriangleright \ast \mathcal{G} \rangle - \langle \delta C, \beta^*(\ast \mathcal{G})\rangle\\[2mm]
&=-\langle \delta A, \overline{\sigma}(\ast \mathcal{G}, B)\rangle - \langle \delta B,d \ast 	\mathcal{G} + A \wedge^\vartriangleright \ast \mathcal{G} \rangle - \langle \delta C, \beta^*(\ast \mathcal{G})\rangle,
\end{align*}
by using $\delta(A \wedge^{\vartriangleright} B) = \delta A \wedge^{\vartriangleright} B + A \wedge^{\vartriangleright} \delta B$, \eqref{AB_1B_2}, \eqref{B_1AB_2} and \eqref{ASD}.\\
And for the third term, we have
\begin{align*}
\langle \delta\mathcal{H}, \ast \mathcal{H}\rangle= &\langle \delta(dC + A \wedge^{\vartriangleright} C + B\wedge^{\left\{ \ ,\  \right\}} B), \ast \mathcal{H} \rangle\\
=&\langle d(\delta C), \ast \mathcal{H}\rangle + \langle \delta A\wedge^\vartriangleright C, \ast \mathcal{H} \rangle + \langle A\wedge^\vartriangleright \delta C, \ast \mathcal{H}\rangle\\
&+ \langle \delta B \wedge^{\left\{ \ ,\  \right\}} B,\ast\mathcal{H}\rangle + \langle B \wedge^{\left\{ \ ,\  \right\}} \delta B, \ast \mathcal{H}\rangle\\[2mm]
=&\langle \delta C, d\ast \mathcal{H}\rangle+(-1)^{d-1} \langle \delta A, \overline{\kappa}(\ast\mathcal{H}, C)\rangle +\langle \delta C,A\wedge^\vartriangleright \ast\mathcal{H}\rangle\\
&- \langle\delta B,\overline{\eta_2}(\ast \mathcal{H}, B) - \langle \delta B, \overline{\eta_1}(\ast \mathcal{H}, B)\\[2mm]
=& (-1)^{d-1}\langle \delta A, \overline{\kappa}(\ast\mathcal{H}, C)\rangle - \langle\delta B,\overline{\eta_2}(\ast \mathcal{H}, B) + \overline{\eta_1}(\ast \mathcal{H}, B)\rangle \\
&+\langle \delta C, d\ast \mathcal{H}+A\wedge^\vartriangleright \ast\mathcal{H}\rangle,
\end{align*}
by using \eqref{AC_1C_2},\eqref{C_1AC_2} and \eqref{B_1B_2C}.
\end{appendix}

\end{document}